\documentclass[10pt]{amsart}
\usepackage[english]{babel}
\usepackage{amssymb}
\usepackage{amsmath}
\usepackage{srcltx}
\usepackage{lscape}

\newcommand{\can}{\overline{\phantom{x}}}

\newtheorem{dummy}{Dummy}

\newtheorem{lemma}[dummy]{Lemma}
\newtheorem{theorem}[dummy]{Theorem}
\newtheorem{proposition}[dummy]{Proposition}
\newtheorem{corollary}[dummy]{Corollary}

\theoremstyle{definition}
\newtheorem{definition}{Definition}

\newtheorem{example}[dummy]{Example}

\newtheorem{remark}[dummy]{Remark}

\newcommand{\ignore}[1]{}

\author{S. Pumpl\"un}
\email{susanne.pumpluen@nottingham.ac.uk}
\address{School of Mathematical Sciences\\
University of Nottingham\\
University Park\\
Nottingham NG7 2RD\\
United Kingdom
}


\begin{document}

\title[Finite nonassociative algebras obtained from skew polynomials]
{Finite nonassociative algebras obtained from skew polynomials  and possible applications to $(f,\sigma,\delta)$-codes}

\maketitle

\begin{abstract}
Let $S$ be a unital ring, $S[t;\sigma,\delta]$ a skew polynomial ring where $\sigma$ is an
injective endomorphism and $\delta$ a left $\sigma$-derivation,
and suppose $f\in S[t;\sigma,\delta]$ has degree $m$ and an invertible leading coefficient.
Using right division by $f$ to define the multiplication,
we obtain unital nonassociative algebras $S_f$ on the set
of skew polynomials in $S[t;\sigma,\delta]$ of degree less than $m$. We study the structure of these algebras.

 When $S$ is a Galois ring and $f$ base irreducible, these algebras yield families of finite unital nonassociative rings $A$, whose
set of (left or right) zero divisors has the form $pA$ for some prime $p$.

 For reducible $f$, the $S_f$ can be employed both to design
linear $(f,\sigma,\delta)$-codes over unital rings and to study their behaviour.
\end{abstract}

%*******************************************************************************************%
%
\section*{Introduction}
%
%*******************************************************************************************%

Let $S$ be a unital  ring.
In the present paper we construct a new class of nonassociative unital rings out of subsets of the skew polynomial ring
$R=S[t;\sigma,\delta]$ where $\sigma$ is an
injective endomorphism and $\delta$ a left $\sigma$-derivation.
Given a polynomial $f\in R=S[t;\sigma,\delta]$ of degree $m$, whose leading coefficient is a unit, it is  well-known
by now (e.g., cf. \cite{Mc}, \cite{JL}, \cite{DO} for commutative $S$) that it is possible to define a right division by $f$:
for all $g(t)\in R$,  there exist  uniquely determined $r(t),q(t)\in R$ with
 ${\rm deg}(r)<m$, such that $g(t)=q(t)f(t)+r(t).$
 What is much less known is the fact that we can take
the additive group $\{g\in R\,|\, {\rm deg} (g)<m \}$ of skew polynomials of degree less than $m$, i.e.
 the canonical representatives of the remainders in $R$ of right division by $f$, and
define a nonassociative unital ring structure $\circ$ on it via $g\circ h=gh \,\,{\rm mod}_r f $.
The resulting nonassociative ring $S_f$, also denoted $S[t;\sigma,\delta]/S[t;\sigma,\delta]f$, is a unital nonassociative algebra over a commutative
subring of $S$. If $f$ is two-sided (also called invariant), i.e. if
$S[t;\sigma,\delta]f$ is a two-sided ideal, then  $S[t;\sigma,\delta]/S[t;\sigma,\delta]f$ is the well-known
associative quotient algebra obtained by factoring out a principal two-sided  ideal. This
generalizes a construction introduced by Petit for the case when $S$ is a division ring
and thus $R=S[t;\sigma,\delta]$ left and right Euclidean \cite{P66}.

The algebras $S_f$ were previously introduced by Petit, but only for the case that $S$ is a division ring,
hence
$S[t;\sigma,\delta]$ left and right Euclidean \cite{P66}. In that setting, they
already appeared in \cite{DO.0}, \cite{DO},  \cite{OS}, and were used in space-time
 block coding, cf.  \cite{SPO12},  \cite{PS15.3}, \cite{PS15.4}.

We present two possible applications:
We first use our algebras to construct new families of finite nonassociative unital rings, especially generalized
nonassociative Galois rings. Generalized nonassociative Galois  rings were
 introduced in \cite{Cons} and investigated in \cite{Cons2}, \cite{Cons3}, \cite{Cons4}.
 They are expected to have wide-ranging applications in coding theory and cryptography \cite{Cons}.

 As a second application, we point out the canonical connection between the algebras $S_f$ and cyclic
 $(f,\sigma,\delta)$-codes. This connection was first mentioned
 in \cite{Pu15.1} for $S$ being a division ring. Well-known results from the literature, e.g.
 on the pseudo-linear map $T_f$ \cite{C} and on polynomials in Ore extensions from \cite{BL} or \cite{L}, are rephrased
 in this setting and put into a nonassociative context.

The paper is organized as follows.
We establish our basic terminology in Section \ref{sec:prel},  define the algebras $S_f$
in Section \ref{sec:2} and investigate their basic structure  in Section \ref{sec:structure}.

The matrix representing left multiplication with $t$ in $S_f$ yields the pseudolinear transformation $T_f$
associated to $f$ defined in \cite{BL} which is discussed in Section \ref{sec:maps}. We generalize
\cite[Theorem 13 (2),  (3), (4)]{LS} and show that
if $S_f$ has no zero divisors then $T_f$ is irreducible, i.e. $\{0\}$ and $S^m$ are the only
$T_f$-invariant left $S$-submodules of $S^m$.

In Section \ref{sec:FCRs}, we assume that $S$ is a finite chain ring. If $f$ is base irreducible then
 $S_f$ is a generalized nonassociative Galois ring. This yields new families of generalized
  nonassociative Galois rings.

We consider the connection between the algebras $S_f$ and cyclic $(f,\sigma,\delta)$-codes,
 in particular skew-constacyclic codes over finite chain rings,
 in Section \ref{sec:codes}:
We rephrase some results (for instance from \cite{B}, \cite{BU14}, \cite{BL},   \cite{JL}, \cite{BSU08}),
 by employing the algebras $S_f$ instead of dealing with cosets in the quotient module
 $S[t;\sigma,\delta]/S[t;\sigma,\delta]f$.
 For instance, the matrix generating a cyclic $(f,\sigma,\delta)$-code $\mathcal{C}\subset S^m$
 represents the right multiplication $R_g$ in $S_f$,
calculated with respect to the basis $1,t,\dots,t^{m-1}$, identifying an element $h=\sum_{i=0}^{m-1}a_it^i$
with the vector $(a_0,\dots,a_{m-1})$, cf. \cite{BL}. This matrix generalizes the circulant matrix from \cite{FG} and
is a control matrix of $\mathcal{C}$. We also show how to obtain
semi-multiplicative maps using their right multiplication. This paper is the starting point for several applications of the algebras $S_f$
to coding theory, e.g. to coset coding, and related areas.
Some are briefly explained in Section 7.

\section{Preliminaries} \label{sec:prel}

\subsection{Nonassociative algebras}\label{subsec:1}

Let $R$ be a unital commutative ring and let $A$ be an $R$-module.
We call $A$ an \emph{algebra} over $R$ if there exists an
$R$-bilinear map $A\times A\mapsto A$, $(x,y) \mapsto x \cdot y$, denoted simply by juxtaposition $xy$,
the  \emph{multiplication} of $A$.
An algebra $A$ is called \emph{unital} if there is
an element in $A$, denoted by 1, such that $1x=x1=x$ for all $x\in A$.
We will only consider unital algebras.

For an $R$-algebra $A$, associativity in $A$ is measured by the {\it associator} $[x, y, z] = (xy) z - x (yz)$.
The {\it left nucleus} of $A$ is defined as ${\rm Nuc}_l(A) = \{ x \in A \, \vert \, [x, A, A]  = 0 \}$, the
{\it middle nucleus}  as ${\rm Nuc}_m(A) = \{ x \in A \, \vert \, [A, x, A]  = 0 \}$ and  the
{\it right nucleus}  as ${\rm Nuc}_r(A) = \{ x \in A \, \vert \, [A,A, x]  = 0 \}$.
${\rm Nuc}_l(A)$, ${\rm Nuc}_m(A)$ and ${\rm Nuc}_r(A)$ are associative subalgebras of $A$.
Their intersection
 ${\rm Nuc}(A) = \{ x \in A \, \vert \, [x, A, A] = [A, x, A] = [A,A, x] = 0 \}$ is the {\it nucleus} of $A$.
${\rm Nuc}(A)$ is an associative
subalgebra of $A$ containing $R1$ and $x(yz) = (xy) z$ whenever one of the elements $x, y, z$ is in
${\rm Nuc}(A)$. The  {\it commuter} of $A$ is defined as ${\rm Comm}(A)=\{x\in A\,|\,xy=yx \text{ for all }y\in A\}$ and
the {\it center} of $A$ is ${\rm C}(A)=\{x\in A\,|\, x\in \text{Nuc}(A) \text{ and }xy=yx \text{ for all }y\in A\}$
 \cite{Sch}.

  An algebra $A\not=0$ over a field $F$ is called a {\it division algebra} if for any $a\in A$, $a\not=0$,
the left multiplication  with $a$, $L_a(x)=ax$,  and the right multiplication with $a$, $R_a(x)=xa$, are bijective.
A division algebra $A$ does not have zero divisors.
If $A$ is a finite-dimensional algebra over $F$, then
$A$ is a division algebra over $F$ if and only if $A$ has no zero divisors.

\subsection{Skew polynomial rings}

Let $S$ be a unital associative (not necessarily commutative) ring, $\sigma$ a ring endomorphism of $S$ and
$\delta:S\rightarrow S$ a \emph{(left) $\sigma$-derivation}, i.e.
an additive map such that
$$\delta(ab)=\sigma(a)\delta(b)+\delta(a)b$$
for all $a,b\in S$, implying $\delta(1)=0$. The \emph{skew polynomial ring} $R=S[t;\sigma,\delta]$
is the set of skew polynomials
$$a_0+a_1t+\dots +a_nt^n$$
with $a_i\in S$, where addition is defined term-wise and multiplication by
$$ta=\sigma(a)t+\delta(a) \quad (a\in S).$$
That means,
$$at^nbt^m=\sum_{j=0}^n a(\Delta_{n,j}\,b)t^{m+j}$$
$(a,b\in S)$, where the map $\Delta_{n,j}$ is defined recursively via
$$\Delta_{n,j}=\delta(\Delta_{n-1,j})+\sigma (\Delta_{n-1,j-1}),$$
with $ \Delta_{0,0}=id_S$, $\Delta_{1,0}=\delta$, $\Delta_{1,1}=\sigma $ and so $\Delta_{n,j}$ is the sum of all polynomials in $\sigma$ and $\delta$
of degree $j$ in $\sigma$ and degree $n-j$ in $\delta$ (\cite[p.~2]{J96} or \cite[p.~4]{BL}).
If $\delta=0$, then $\Delta_{n,j}=\sigma^n$.

$S[t;\sigma]=S[t;\sigma,0]$ is called a \emph{twisted polynomial ring} and
$S[t;\delta]=S[t;id,\delta]$ a \emph{differential polynomial ring}.
For $\sigma=id$ and $\delta=0$, we obtain the usual ring of left polynomials $S[t]=S[t;id,0]$.

 For $f=a_0+a_1t+\dots +a_nt^n$ with $a_n\not=0$ define ${\rm deg}(f)=n$ and ${\rm deg}(0)=-\infty$.
Then ${\rm deg}(gh)\leq{\rm deg} (g)+{\rm deg}(h)$ (with equality if $h$ has an invertible leading coefficient,
or $g$ has an invertible leading coefficient and $\sigma$ is injective, or if $S$ is a division ring).
 An element $f\in R$ is \emph{irreducible} in $R$ if it is not a unit and  it has no proper factors, i.e if there do not exist $g,h\in R$ with
 ${\rm deg}(g),{\rm deg} (h)<{\rm deg}(f)$ such
 that $f=gh$.

Suppose  $D$ is a division ring. Then $R=D[t;\sigma,\delta]$ is a left principal ideal domain (i.e.,  every left ideal in $R$ is of the form $Rf$)  and
there is a right division algorithm in $R$  \cite[p.~3]{J96}: for all $g,f\in R$, $g\not=0$, there exist unique $r,q\in R$,
and ${\rm deg}(r)<{\rm deg}(f)$, such that
$$g=qf+r$$
(cf. Jacobson \cite{J96} and Petit \cite{P66}, note that Jacobson calls what we call right a left division
algorithm and vice versa.).
If $\sigma$  is a ring automorphism then $R=D[t;\sigma,\delta]$ is a left and right principal ideal domain (a PID)
 \cite[p.~6]{J96} and
there is also a left division algorithm in $R$ \cite[p.~3 and Prop. 1.1.14]{J96}.

%%%%%%%%%%%%%%%%%%%%%%%%%%%%%%%%%%%%%%%%%%%%%%%%%%%%%%%%%%%%%%%%%%%%%%%

\section{Nonassociative rings obtained from skew polynomials rings} \label{sec:2}

%%%%%%%%%%%%%%%%%%%%%%%%%%%%%%%%%%%%%%%%%%%%%%%%%%%%%%%%%%%%%%%%%%%%%%%%%

 From now on, let $S$ be a unital ring and $S[t;\sigma,\delta]$ a skew polynomial ring where $\sigma$ is injective.
 $S[t;\sigma,\delta]$ is generally neither a left nor a right Euclidean ring (unless $S$ is a division ring).
Nonetheless, we can still perform a left and right division by a polynomial $f \in R=S[t;\sigma,\delta]$, if
$f(t)=\sum_{i=0}^{m}d_it^i$ has an invertible leading coefficient $LC(f)=d_m$
(this was already observed for twisted polynomial rings and special cases of $S$ and assuming $\sigma\in{\rm Aut}(S)$ for instance in
 \cite[p.~391]{Mc}, \cite[p.~4]{JL}, \cite[3.1]{DO}):

\begin{proposition} \label{prop:mainr}
Let $f(t)\in R=S[t;\sigma,\delta]$ have degree $m$ and an invertible leading coefficient.
\\ (i) For all $g(t)\in R$ of degree $l\geq m$,  there exist  uniquely determined $r(t),q(t)\in R$ with
 ${\rm deg}(r)<{\rm deg}(f)$, such that
$$g(t)=q(t)f(t)+r(t).$$
(ii) Assume $\sigma\in{\rm Aut}(S)$.
Then for all $g(t)\in R$ of degree $l\geq m$,  there exist  uniquely determined $r(t),q(t)\in R$
with ${\rm deg}(r)<{\rm deg}(f)$, such that
$$g(t)=f(t)q(t)+r(t).$$
\end{proposition}

\begin{proof}
(i)
Let $f(t)=\sum_{i=0}^{m}d_it^i$ and $g(t)=\sum_{i=0}^{l}s_it^i$ be two skew polynomials
in $R$ of degree $m$ and $l$. Suppose that $l>m$ and that the leading coefficient $LC(f)=d_m$ of $f$
is invertible. Since
$1=\sigma(d_m d_m^{-1})=\sigma(d_m)\sigma(d_m^{-1})$, we know that $\sigma(d_m)$ and thus
$\sigma^j(d_m)$ is invertible for any integer $j\geq 0$.
Now
$$g(t)-s_l \sigma^{l-m}(d_m^{-1})t^{l-m}f(t)=g(t)-s_l \sigma^{l-m}(d_m^{-1})t^{l-m} (d_mt^m+\sum_{i=0}^{m-1}d_it^i)$$
$$=g(t)-s_l \sigma^{l-m}(d_m^{-1})t^{l-m} d_mt^m -\sum_{i=0}^{m-1}s_l \sigma^{l-m}(d_m^{-1})t^{l-m}d_it^i$$
$$=g(t)-s_l \sigma^{l-m}(d_m^{-1})(\sum_{j=0}^{l-m}\Delta_{l-m,j}(d_m)t^j)t^m-\sum_{i=0}^{m-1}s_l \sigma^{l-m}(d_m^{-1})
(\sum_{j=0}^{l-m}\Delta_{l-m,j}(d_m)t^j)t^i$$
$$=g(t)-s_l \sigma^{l-m}(d_m^{-1})\Delta_{l-m,l-m}(d_m)t^l$$
$$-s_l \sigma^{l-m}(d_m^{-1})
\sum_{j=0}^{l-m-1}\Delta_{l-m,j}(d_m)t^{j+m}-\sum_{i=0}^{m-1}\sum_{j=0}^{l-m}s_l \sigma^{l-m}(d_m^{-1})\Delta_{l-m,j}(d_j)t^{i+j}$$
$$=g(t)-s_l t^l$$
$$-s_l \sigma^{l-m}(d_m^{-1})
\sum_{j=0}^{l-m-1}\Delta_{l-m,j}(d_m)t^{j+m}-\sum_{i=0}^{m-1}\sum_{j=0}^{l-m}s_l \sigma^{l-m}(d_m^{-1})\Delta_{l-m,j}(d_j)t^{i+j}.$$
Note that we used that $\Delta_{l-m,l-m}(d_m)=\sigma^{l-m}(d_m)$ in the last equation. Therefore the polynomial
$g(t)-s_l \sigma^{l-m}(d_m)t^{l-m}f(t)$
has degree $<l$. By iterating this argument, we find $r,q\in R$ with
 ${\rm deg}(r)<{\rm deg}(f)$, such that
$$g(t)=q(t)f(t)+r(t).$$
To prove uniqueness of $q(t)$ and the remainder $r(t)$, suppose we have
$$g(t)=q_1(t)f(t)+r_1(t)=q_2(t)f(t)+r_2(t).$$
Then $(q_1(t)-q_2(t))f(t)=r_2(t)-r_1(t)$. If $q_1(t)-q_2(t)\not=0$ and observing that $f$ has invertible leading
coefficient such that $\sigma(d_m)^j$ cannot be a zero divisor for any positive $j$, we conclude that the degree of the left-hand side of the
equation is greater than ${\rm deg}(f)$ and the degree of $r_2(t)-r_1(t)$ is less than ${\rm deg}(f)$, thus
$q_1(t)=q_2(t)$ and  $r_1(t)=r_2(t)$.
\\ (ii)  The proof is along similar lines as the one of (i), using that the polynomial
$g(t)-f(t)  \sigma^{-m}(s_l) \sigma^{-m}(d_m^{-1})t^{l-m}$
has degree $<l$ and iterating this argument. The uniqueness of $q(t)$ and the remainder is proved analogously as in (i).
\end{proof}

In the following,  we always assume that
  $$f(t)\in S[t;\sigma,\delta] \text{\emph{ has degree} } m>1 \text{\emph{ and an invertible leading coefficient
  } }LC(f).$$
Let ${\rm mod}_r f$ denote the remainder of right division by $f$ and
 ${\rm mod}_l f$ the remainder of left division by $f$.
Since the remainders are uniquely determined, the skew polynomials of degree less that $m$ canonically represent the
elements of the left  $S[t;\sigma,\delta]$-module $S[t;\sigma,\delta]/ S[t;\sigma,\delta]f$ and when $\sigma\in{\rm Aut}(S)$,
for the right $S[t;\sigma,\delta]$-module
$S[t;\sigma,\delta]/ fS[t;\sigma,\delta]$.

 \begin{definition} \label{def:Petit}
 Suppose $f(t)=\sum_{i=0}^{m}d_it^i\in R=S[t;\sigma,\delta]$.
\\ (i) The additive group $\{g\in R\,|\, {\rm deg}(g)<m\}$  together with the multiplication
$$g\circ h=gh \,\,{\rm mod}_r f$$
defined for all $g,h\in R$ of degree less than $m$, is a unital nonassociative ring $S_f$ also denoted by $R/Rf$.
\\ (ii) Suppose $\sigma\in{\rm Aut}(S)$. Then the additive group $\{g\in R\,|\, {\rm deg}(g)<m\}$  together with the multiplication
$$g\diamond h=gh \,\,{\rm mod}_l f $$
defined for all $g,h\in R$ of degree less than $m$, is a unital nonassociative ring $\,_fS$ also denoted by $R/fR$.
\end{definition}

$S_f$ and $\,_fS$ are unital algebras over $S_0=\{a\in S\,|\, ah=ha \text{ for all } h\in S_f\}$,
which is a commutative subring of $S$.
If $S$ is a division ring, Definition \ref{def:Petit} is Petit's algebra construction  \cite{P66}  and $S_0$  is a
subfield of $S$.
In the following, we therefore call the algebras $S_f$  \emph{Petit algebras}.

\begin{remark} \label{rem:1}
(i) Let $g,h\in R$ have degrees less than $m$.
  If ${\rm deg}(gh)<m$ then the multiplication $g\circ h$ in $S_f$ and $g\diamond h$ in $\,_fS$ is the usual multiplication of  polynomials in $R$.
\\ (ii) If $Rf$ is a two-sided ideal in $R$ (i.e. $f$ is \emph{two-sided}, also called \emph{invariant}) then $S_f$ is
the   associative quotient algebra
 obtained by factoring out the ideal generated by a two-sided $f\in S[t;\sigma,\delta]$.
\\ (iii) If $f\in S[t;\sigma,\delta]$ is reducible then $S_f$ contains zero divisors:
if $f(t)=g(t)h(t)$ then $g(t)$ and $h(t)$ are zero divisors in $S_f$.
The argument leading up to \cite[Section 2., (6)]{P66} shows that  if $S$ is a division ring, then
$S_f$ has no zero divisors if and only if $f$ is irreducible, which is in turn equivalent to
$S_f$ being a right division ring (i.e., right multiplication $R_h$ in $S_f$ is bijective for all $0\not=h\in S_f$).

However, for general rings $S$ it can happen that $S_f$ has zero divisors,  even when $f$ is irreducible.
\\ (iv) For all invertible $a\in S$ we have $S_f=S_{af}$, so that without loss of generality it suffices to only
consider monic polynomials in the construction.
\end{remark}

It suffices to consider the algebras $S_f$, since we have the following canonical anti-automorphism
(cf. \cite[(1)]{P66} when $S$ is a division ring, the proof is analogous):

 \begin{proposition}
 Let $f\in R=S[t;\sigma,\delta]$ have an invertible leading coefficient and let $\sigma\in{\rm Aut}(S)$.
 The canonical anti-automorphism
 $$\psi: S[t;\sigma,\delta]\rightarrow S^{op}[t;\sigma^{-1},-\delta\circ\sigma^{-1}],$$
 $$\psi(\sum_{k=0}^{n}a_kt^k)=\sum_{k=0}^{n}(\sum_{i=0}^{k}\Delta_{n,i}(a_k))t^k$$
 between the skew polynomial rings $S[t;\sigma,\delta]$ and $ S^{op}[t;\sigma^{-1},-\delta\circ\sigma^{-1}]$ induces an anti-automorphism
 between the rings $$S_f=S[t;\sigma, \delta]/ S[t;\sigma,\delta]f$$
  and
 $$\,_{\psi(f)}S=S^{op}[t;\sigma^{-1},-\delta\circ\sigma^{-1}]/\psi(f) S^{op}[t;\sigma^{-1},-\delta\circ\sigma^{-1}].$$
\end{proposition}

Note that if $\delta=0$ and $\sigma\in{\rm Aut}(S)$, we have
$$\psi(\sum_{k=0}^{n}a_kt^k)=\sum_{k=0}^{n}\sigma^{-k}(a_k)t^k.$$

%%%%%%%%%%%%%%%%%%%%%%%%%%%%%%%%%%%%%%%%%%%%%%%%%%%%%%%%%%%%%%%%%%%%%%%

\section{Some structure theory} \label{sec:structure}

%%%%%%%%%%%%%%%%%%%%%%%%%%%%%%%%%%%%%%%%%%%%%%%%%%%%%%%%%%%%%%%%%%%%%%%%%

\subsection{}
In the following, let $f\in R=S[t;\sigma,\delta]$ be monic of degree $m>1$ and $\sigma$ injective.
 When
 $S$ is a division ring,  the structure of $S_f$ is extensively investigated in  \cite{P66}. For instance,
 if $S$ is a division ring and the $S_0$-algebra $S_f$ is finite-dimensional, or free of finite rank as a right module
 over its right nucleus, then $S_f$ is a division algebra if and only if $f(t)$ is irreducible \cite[(9)]{P66}.

Some of the results in \cite{P66} carry over  to our more general setting:

\begin{theorem}  \label{thm:main1}
 (i) $S_f$ is a free left $S$-module  of rank $m$ with basis $t^0=1,t,\dots,t^{m-1}$.
 \\ (ii) $S_f$ is associative if and only if $f$ is two-sided.
\\ (iii) If $S_f$ is not associative then
$$S\subset{\rm Nuc}_l(S_f),\,\,S\subset{\rm Nuc}_m(S_f)$$
and
$$\{g\in R\,|\, {\rm deg}(g)<m \text{ and }fg\in Rf\}= {\rm Nuc}_r(S_f).$$
When $S$ is a division ring, the inclusions become equalities.
\\ (iv)
We have $t\in {\rm Nuc}_r(S_f)$, if and only if
the powers of $t$ are associative, if and only if $t^mt=tt^m$ in $S_f$.
\\ (v)  If $S$ is a division ring and $S_f$ is not associative then
$$C(S_f)=S_0.$$
 (vi) Let $f(t)=\sum_{i=0}^{m}d_it^i\in S[t;\sigma]$ with $d_0$ invertible.
If the endomorphism  $L_t$ which is the left multiplication by $t$ as defined in Section \ref{subsec:1} is
 surjective then $\sigma$ is surjective.
  In particular, if $S$ is a division ring and $f$ irreducible, then $L_t$ surjective implies $\sigma$ surjective.

  Moreover, if $\sigma$ is bijective then $L_t$ is surjective.
\end{theorem}

\begin{proof}
(i) is clear.
\\
(ii) If $f$ is two-sided, $S_f$ is clearly associative.
Conversely,
if $S_f$ is associative then $S_f={\rm Nuc}_r(S_f)=\{g\in R\,|\, {\rm deg}(g)<m \text{ and } fg\in Rf\}.$ Thus $t\in {\rm Nuc}_r(S_f)$
and also $S\subset {\rm Nuc}_r(S_f)$. This means $f(t)t\in Rf(t)$ and for all $a\in S$, also
$f(t)a=g(t)f(t)$ for a suitable $g(t)\in R$. Comparing degrees (recall we  assume $f$ to have an invertible
leading coefficient) we see that $g(t)=b\in S$, so we get
$f(t)t\in Rf(t)$ and for all $a\in S$, also
$f(t)a=bf(t)$ for a suitable $b\in R$. Thus $f$ is invariant, i.e. two-sided.
\\ (iii)  The proof of the first two inclusions and that
$\{g\in R\,|\, {\rm deg}(g)<m \text{ and }fg\in Rf\}\subset {\rm Nuc}_r(S_f)$ is similar to \cite[(2)]{P66} (which proves the result for $S$ being a division ring),
 as this inclusion does not need $S$ to be a division ring.
For instance, for $a\in {\rm Nuc}_l(S_f)=\{a\in S_f\,|\, [a,b,c]=0
\text{ for all }b,c\in S_f\} $ we have
$[a,b,c]=0$ if and only if $pfc=0$ for some $p\in R$. If $a$ has degree 0 then $p=0$ as observed in \cite[(2)]{P66} so
$S\subset {\rm Nuc}_l(S_f)$.
It remains to show that ${\rm Nuc}_r(S_f)\subset\{g\in R\,|\, {\rm deg}(g)<m \text{ and }fg\in Rf\}$ \cite{CB}:
Let $b,c,d\in R$ have degree less than $m$. Write $bc=q_1f+r_1$, $cd=q_2f+r_2$ with $q_i,r_i\in R$ uniquely determined of degree
smaller than $m$. A straightforward calculation as in \cite[(2)]{P66} shows that in $S_f$ we thus have
$(bc)d=b(cd)$ if and only if $q_1fd\, {\rm mod}_r f=0$ if and only if $q_1fd\in Rf$.

Let now $d\in{\rm Nuc}_r(S_f)$ and choose $b,c\in R$ with invertible leading coefficient such that
${\rm deg}(b)+{\rm deg}(c)=m$, so that ${\rm deg}(bc)=m$. Write $bc=q_1f+r_1$. Then ${\rm deg}(q_1f)={\rm deg}(q_1)+m$.
But here $bc=q_1f+r_1$ also means ${\rm deg}(q_1)=0$, so $q_1\in S$ is non-zero. The leading coefficient of
$bc$ is $LC(b)LC(c)$ and
the leading coefficient of $q_1f$ is $q_1 $.
Therefore
$q_1=LC(b)\sigma^l(LC(c))$ is invertible in $S$.
Since  $d\in{\rm Nuc}_r(S_f)$ implies $q_1fd\in Rf$, this yields $fd\in Rf$.
\\ (iv) If $ft\in Rf$ then $t\in {\rm Nuc}_r(S_f)$ by (iii),
hence $t,\dots, t^{m-1}\in {\rm Nuc}_r(S_f)$, and so $[t^i,t^j,t^k]=0$ for all $i,j,k<m$, meaning the powers of $t$
are associative. In particular, this implies $[t,t^{m-1},t]=0$, that is $t^mt=tt^m$.
A careful analysis of the proof of \cite[(5)]{P66} shows that the other implications can be proved analogously as in
\cite[(5)]{P66}, also when also holds when $S$ is not a division algebra, since we still have that
we have $[t^i,t^j,t^k]=0$ for all $i,j,k<m$ with $i+j<m$ analogously as in \cite[6]{P66}.
\\ (v) We have $C(S_f)={\rm Comm}(S_f)\cap {\rm Nuc}(S_f)={\rm Comm}(S_f)\cap S=S_0$.
\\ (vi)
 If $d_0$ is invertible and $\delta=0$ then $L_t$ surjective implies $\sigma$ surjective:
For $u=\sum_{i=0}^{m-1}u_it^i\in S_f$, we have (using the multiplication in $S_f$)
$$L_t(u)=\sum_{i=0}^{m-2}\sigma(u_i)t^{i+1}+\sigma(u_{m-1})t^m=\sum_{i=0}^{m-2}\sigma(u_i)t^{i+1}+
\sigma(u_{m-1})\sum_{i=0}^{m-1}d_it^i.$$
Suppose $L_t$ is surjective, then given any $b\in S$, there is $u\in S_f$ such that $L_t(u)=b$.
Comparing the constants in this equation, we obtain that for all $b\in S$ there is $u_{m-1}\in S$
such that $\sigma(u_{m-1})=bd_0$, i.e. for all $c\in S$
 there is $u_{m-1}\in S$ such that $\sigma(u_{m-1})=c$ \cite{CB}.

 The statement that if $S$ is a division ring and $f$ irreducible then $L_t$ is surjective
 implies $\sigma$ surjective is \cite[Section 2., (6)]{P66} and follows as a special case now.

 If $\sigma$ is bijective then $L_t$ is surjective: Let $g=\sum_{i=0}^{m-1}g_it^i$.
 Define $u_{m-1}=\sigma^{-1}(g_0d_0^{-1})$
  and $u_{i-1}=\sigma^{-1}(g_i)-u_{m-1}\sigma^{-1}(d_i)$.
 Then $L_t(u)=g$ \cite{CB}.
\end{proof}

Recall that the largest subalgebra of $R=S[t;\sigma,\delta]$ in which $Rf$ is a two-sided ideal is
the \emph{idealizer}  $I(f)=\{g\in R\,|\, fg\in Rf\}$ of $Rf$. The \emph{eigenring}
of $f$ is then defined as the quotient $E(f)=I(f)/Rf$. The
eigenring $E(f)=\{g\in R\,|\, {\rm deg}g<m \text{ and } fg\in Rf\}$ equals the right nucleus ${\rm Nuc}_r(S_f)$ by Theorem \ref{thm:main1} (iii)
(or see \cite[(2)]{P66} if $S$ is a division algebra)
 which, as the right nucleus, is an associative subalgebra of $S_f$, cf. Section \ref{subsec:1}.
 More precisely, the multiplication $\circ$ in $S_f$ makes ${\rm Nuc}_r(S_f)$ into an associative algebra
 which equals the associative quotient ring $E(f)$ equipped with the canonical multiplication induced on it by the multiplication
 on the ring $I(f)\subset R$.
 When $S$ is a division ring, non-trivial zero divisors in $E(f)={\rm Nuc}_r(S_f)$ correspond to factors of $f$:

\begin{proposition} (\cite[Proposition 4]{G})\label{prop:nuctest}
Let $S$ be a division ring  and $f\in R=S[t;\sigma,\delta]$.
\\ (i) Let $uv=0$ for some non-zero $u,v\in E(f)$, then the greatest common right divisor
$gcrd(f,u)$  is a non-trivial right factor of $f$.
($v\in R$ is  the \emph{greatest common right divisor of $f$ and $u$}, written $gcrd(f,u)=v$,
 if there are $s,d\in R$ such that $sf+du=v$.)
\\ (ii) Let $f\in R$ be \emph{bounded} (i.e., there exists $0\not=f^*\in R$ such that $Rf^*=f^*R$
is the largest two-sided ideal of $R$ contained in $Rf$) and $\sigma$ be an automorphism.
Then $f$ is irreducible if and only if
$E(f)={\rm Nuc}_r(S_f)$ has no non-trivial zero divisors.
\end{proposition}

\begin{remark} \label{rem:7}
 Let $S$ be a division ring.
\\ (i) If $f$ is irreducible then ${\rm Nuc}_r(S_f)$ is an associative division algebra \cite[p.~17-19]{G}.
\\ (ii) Effective algorithms to compute ${\rm Nuc}_r(S_f)$  for $f\in\mathbb{F}_q(x)[t;\sigma]$ and
$f\in\mathbb{F}_q(x)[t;\delta]$   can be found in \cite{GZ}, for $R=\mathbb{F}_q[t;\sigma]$  in \cite{G0}, \cite{R}.
 Proposition \ref{prop:nuctest} is also employed for linear differential operators in \cite{Si},
  to factorize skew polynomials for $S=\mathbb{F}_q$ in \cite{G0}
and for $S=\mathbb{F}_q(x)$ in \cite{GZ}, \cite{GLN},  \cite{G}, without relating it to the algebras $S_f$ however.
\end{remark}

\begin{proposition} \label{prop:skewcodemain}
 Let  $f\in R=S[t;\sigma,\delta]$.
\\ (i) Every right divisor $g$ of $f$ of degree $<m$  generates a principal left ideal in $S_f$.

 All non-zero  left ideals in $S_f$ which contain a polynomial $g$
of minimal degree with invertible leading coefficient are principal ideal generated by $g$, and $g$ is a right divisor of $f$ in $R$.
\\ (ii) Each principal left ideal generated by a right divisor of $f$ is
an $S$-module which is isomorphic to a submodule of $S^m$.
\\ (iii) If $f$ is irreducible, then $S_f$ has no non-trivial principal left ideals
which contain a polynomial of minimal degree with invertible leading coefficient.
\end{proposition}

The proof is straightforward. If there is no  polynomial $g$
of minimal degree with invertible leading coefficient in a non-zero left ideal, then the ideal need not be principal, see
\cite[Theorem 4.1]{JL} for examples.

\begin{theorem} \label{thm:main2}
Let $f\in R=S[t;\sigma]$.
 \\ (i) The commuter ${\rm Comm}(S_f)=\{g\in S_f\,|\, gh=hg \text{ for all } h\in S_f\}$ contains the set
 $$\{\sum_{i=0}^{m-1}a_it^i \,|\, a_i\in {\rm Fix}(\sigma)\text{ and } ca_i=a_i\sigma^i(c)\text{ for all } c\in S \}.$$
If $t$ is left-invertible in $S_f$ and $S$ a division ring, the two sets are equal.
\\ (ii) ${\rm Fix}(\sigma)\cap  C(S) \subset S_0= {\rm Comm}(S_f)\cap S$.
If $t$ is left-invertible and $S$ a division ring, the two sets are equal.
\end{theorem}

\begin{proof}
(i) and (iii) are straightforward  calculations; both generalize
 \cite[(14), (15)]{P66}.
\\ (ii) follows from (i):
$S_0=\{a\in S\,|\, ah=ha \text{ for all } h\in S_f\}= {\rm Comm}(S_f)\cap S$ and
${\rm Fix}(\sigma)\cap C(S) \subset {\rm Comm}(S_f)\cap S= S_0$. If $t$ is left-invertible,
the two sets are equal.
\end{proof}

\begin{remark} \label{re:left-inv}
For $f(t)=\sum_{i=0}^{m}d_it^i\in S[t;\sigma]$ monic,
$t$ is left-invertible if and only if $d_0$ is left-invertible.
One direction is a  simple degree argument (suppose there are $g,h\in S_f$ with $gt=hf+1$, then
compare the constant terms of both sides).
Conversely, if $d_0$ is left-invertible then $t$ is left-invertible (say, $h_0d_0=1$, choose $h=-h_0$ and define
$g(t)=\sum_{i=0}^{m-1}hd_{i+1}t^i$ to get $gt=hf+1$).
Thus if $f$ is irreducible (hence $d_0\not=0$) and $S$ a division ring then
$t$ is always left-invertible and $S_0={\rm Fix}(\sigma)\cap {\rm Comm}(S)$.
\end{remark}

%%%%%%%%%%%%%%%%%%%%%%%%%%%%%%%%%%%%%%%%%%%%%%%%%%%%%%%%%%%%%%%%
\subsection{When $S$ is an integral domain}
%%%%%%%%%%%%%%%%%%%%%%%%%%%%%%%%%%%%%%%%%%%%%%%%%%%%%%%%%%%%%%%%%%

 In this section, let  $S$ be a commutative integral domain  with quotient field $K$, $f$ be monic and $\sigma$ injective
 as before. Then
$\sigma$ and $\delta$ canonical extend to $\sigma$ and $\delta$ to $K$ via
$$\sigma(\frac{a}{b})=\frac{\sigma(a)}{\sigma(b)},$$
$$\delta(\frac{a}{b})=\frac{\delta(a)}{b}-\frac{ \sigma(\frac{a}{b}) \delta(b)}{b}$$
for all $a,b\in S$, $b\not=0$.

\begin{proposition} \label{prop:domain}
 Let $S$ be an integral domain with quotient field $K$,
 $f\in S[t;\sigma,\delta]$ and let $S_f=S[t;\sigma,\delta]/S[t;\sigma,\delta]f$.
\\ (i)  $S_f\otimes K\cong K[t;\sigma,\delta]/K[t;\sigma,\delta]f$
again is a Petit algebra.
\\ (ii)  If $f$ is irreducible in $K[t;\sigma,\delta]$, then $S_f$
has no zero divisors.
\end{proposition}

\begin{proof}
(i):  The isomorphism is clear by \cite[3]{P66}.
\\ (ii): By (i), we have $S_f\otimes K\cong K[t;\sigma,\delta]/K[t;\sigma,\delta]f$.
Since $f(t)$ is irreducible in $K[t;\sigma,\delta]$ and $K$ is field,
$K[t;\sigma,\delta]/K[t;\sigma,\delta]f$ is a Petit algebra such that $R_h$ is bijective and $L_h$ is injective,
for all $0\not=h\in S_f$ \cite[Section 2., (6)]{P66}. This implies that it does not have any zero divisors, and so neither does
 $S_f=S[t;\sigma,\delta]/S[t;\sigma,\delta]f$.
\end{proof}

\begin{example} \label{ex:cyclic}
Nonassociative cyclic division algebras were introduced by Sandler \cite{S} and  studied in \cite{S12}
(to be precise, \cite{S12} looks at their opposite algebras).
We generalize their definition  (see \cite{OS} for the associative set-up):

Let $S/S_0$ be an extension of commutative rings, $\sigma\in {\rm Aut}(S)$  and $G=\langle \sigma\rangle$ a finite cyclic
group of order
$m$ acting on $S$ such that the action is trivial on $S_0$. For any $c\in S$, the \emph{generalized (associative or
nonassociative) cyclic algebra}  $A=(S/S_0,\sigma,c)$ is the $m$-dimensional $S$-module
$A=S \oplus St \oplus St^2 \oplus \dots\oplus St^{m-1}$
 where multiplication is given by the following relations
for all $a,b\in S, 0 \leq i,j, <m$, which then are extended linearly to all elements of $A$:
\[
 (at^i)(bt^j) =
  \begin{cases}
   a\sigma^i(b)   t^{i+j} & \text{if } i+j < m, \\
     a \sigma^i(b)  t^{(i+j)-m}c & \text{if } i+j \geq m,
  \end{cases}
\]
If $\sigma\in {\rm Aut}(S)$, then $(S/S_0,\sigma,c)=S_f$ for $f(t)=t^m-c \in S[t;\sigma]$
and $S_0={\rm Fix}(\sigma)$. If $c \in S \setminus S_0$, the algebra $(S/S_0,\sigma,c)$ has nucleus  $S$  and center $S_0$.

Suppose $S_0$ and $S$ are integral domains with quotient fields $F$ and $K$. Canonically extend
$\sigma$ to an automorphism $\sigma:K\to K$, then
if $m$ is prime, $(S/S_0,\sigma,c)=S_f$ has no zero divisors for any choice of $c \in S \setminus S_0$
(since then $(K/F,\sigma,c)$ always is a nonassociative cyclic division algebra and contains $S_f$).

Generalized associative cyclic algebras are used  in \cite{DO}, generalized nonassociative cyclic algebras
in \cite{Pu15}.
\end{example}

%%%%%%%%%%%%%%%%%%%%%%%%%%%%%%%%%%%%%%%%%%%%%%%%%%%%%%%%%%%%%%%%%%%%%%%%%%%%%%%%%%%%%%%%%%%%%%%%

\section{Pseudolinear maps} \label{sec:maps}

%%%%%%%%%%%%%%%%%%%%%%%%%%%%%%%%%%%%%%%%%%%%%%%%%%%%%%%%%%%%%%%%%%%%%%%%%%%%%%%%%%%%%%%%%%%%%%%

Let $\sigma$ be injective  and $f=\sum_{i=0}^{m}d_it^i\in S[t;\sigma,\delta]$ be  a monic skew polynomial of degree $m>1$.
By Theorem \ref{thm:main1},
$S_f$ is a free left $S$-module with $S$-basis $1,t,\dots,t^{m-1}$. We identify an element $h\in S_f$,
$h(t)=\sum_{i=0}^{m-1}a_it^i$ with the vector $(a_0,\dots,a_{m-1})\in S^m$.

 Right multiplication with $0\not=h\in S_f$ in $S_f$,
$R_h:S_f\longrightarrow S_f,$ $p\mapsto ph$, is an $S$-module endomorphism \cite{P66}.
After expressing  $R_h$ in matrix form
with respect to the $S$-basis $1,t,\dots, t^{m-1}$ of $S_f$, the map
$$\gamma: S_f \to {\rm End}_K(S_f), h\mapsto R_h$$
induces an injective $S$-linear map
$$\gamma: S_f \to {\rm Mat}_m(S), h\mapsto R_h \mapsto Y.$$

Left multiplication $L_h:S_f\longrightarrow S_f,$ $p\mapsto hp$ is an
$S_0$-module endomorphism. If we consider $S_f$ as a right ${\rm Nuc}_r(S_f)$-module then $L_h$ is a
${\rm Nuc}_r(S_f)$-module endomorphism.

For a two-sided  $f$, $\gamma$ is the right regular representation and
$\lambda$ is the left regular representation of the  associative algebra $S_f$.

If $S$ is a commutative ring and $\det(\gamma(h))=\det Y=0$, then $h$ is a right zero divisor in $S_f$.
Moreover, $S_f$ is a division algebra if and only if $\gamma(h)$ is an invertible matrix for every nonzero $h \in S_f$.

\begin{remark}\label{re:I}
(i)  In \cite{FG}, where $S$ is a finite field and
  $f(t)=t^n-a$, $\delta=0$, $\gamma(h)=Y$ is the \emph{circulant matrix}  $M_a^\theta $.
\\ (ii) If $S$ is not commutative, but contains a suitable commutative subring, it is  still
possible to define a matrix representing left or right multiplication
in the $S_0$-algebra $S_f$ where the entries of the matrix lie in a commutative subring of $S$ which strictly contains
$S_0$ and which displays the same behaviour as above. This is a particularity of Petit's algebras, and
not always possible for nonassociative algebras in general. It
 reflects the fact that the left nucleus of $S_f$ always contains $S$ (and thus is
 rather `large') and that also the right nucleus
 may contain $S$ or subalgebras of $S$, depending on the $f$ used in the construction.

 For instance, this is the case (and was used when designing fast-decodable space-time block codes, e.g. in \cite{PS15.3},
 \cite{PS15.4}, \cite{Pu16.1}) when $S$ is a cyclic division
algebra $S=(K/F,\rho,c)$ of degree $n$ and $f(t)=t^m-d\in S[t;\sigma]$, with $\sigma$ suitably chosen.
The $m\times m$ matrix $\gamma(h)=Y$ consequently has its entries in $(K/F,\rho,c)$.
We can then substitute each entry in the matrix, which has the form $\sigma^i(x)$ for some $x\in (K/F,\rho,c)$,
perhaps timed with the scalar $d$, with an $n\times n$ matrix: take the matrix of the
right regular representation of $x$ over $K$ in $(K/F,\rho,c)$, apply $\sigma^i$ to each of its entries
and using scalar multiplication by $d$ if applicable.
We obtain an $mn\times mn$ matrix $X$ with entries in the field $K$, which still represents right multiplication
with an element in $S_f$, but now written with respect to the canonical $K$-basis
$1,\dots,e, t,\dots, et,\dots, e^{n-1}t^{m-1}$ of $S_f$, $1,e,\dots,e^{n-1}$ being the canonical basis of $(K/F,\rho,c)$.
Again $\det X=0$  implies
that $h$ is a  zero divisor in $S_f$, and were $S_f$ is a division algebra if and only if $X$ is invertible
for every non-zero $h\in S_f$.
The interested reader is referred to \cite{SPO12}, \cite{PS15.3}, \cite{PS15.4}, \cite{Pu16.1} for the details which would be
beyond the scope of this paper.
\end{remark}

Let
\[C_f = \left[ \begin{array}{ccccc}
0 & 1 & 0 & \cdots & 0 \\
0 & 0 & 1 & 0 & \cdots \\
\vdots & \vdots & \vdots & \ddots & \vdots \\
0 & 0 & 0 & 0 & 1 \\
-d_{0} & -d_{1} &  & \cdots & -d_{m-1} \end{array} \right] \]
be the \emph{companion matrix} of $f$. Then
$$T_f:S^m\longrightarrow S^m,\quad T_f(a_1,\dots,a_m)=(\sigma(a_1),\dots,\sigma(a_m))C_f+(\delta(a_1),\dots,\delta(a_m))$$
is a $(\sigma,\delta)$-pseudolinear transformation on the left $S$-module $S^m$, i.e. an additive map such that
$$T_f(ah)=\sigma(a)T_f(h)+\delta(a)$$
 for all $a\in S$, $h\in S^m$. $T_f$ is called the \emph{pseudolinear transformation
associated to} $f$ \cite{BL} and we
can translate some results on $T_f$ (e.g., see \cite{L}) to our nonassociative context.
For $h=\sum_{i=0}^{n}a_it^i\in S[t;\sigma,\delta]$ we define
$$h(T_f)=\sum_{i=0}^{n}a_i T_f^i.$$

\begin{theorem}\label{thm:leftmult}
(i)
The pseudolinear transformation $T_f$ is the left multiplication $L_t:S_f\longrightarrow S_f,
h\mapsto th$ with $t$ in $S_f$,
calculated with respect to the basis $1,t,\dots,t^{m-1}$, identifying an element $h=\sum_{i=0}^{m-1}a_it^i$
with the vector $(a_0,\dots,a_{m-1})$:
$$L_t(h)=T_f(h)$$
for all $h\in S_f$.
\\ (ii) We have $L_t^i(h)=L_{t^i}(h)$ for all $h\in S_f$.
\\ (iii) Left multiplication $L_h$ with  $h\in S_f$ is given by
$$L_h=h(T_f)=\sum_{i=0}^{n}a_i T_f^i,$$
or equivalently by
$$L_h=h(L_t)=\sum_{i=0}^{n}a_i L_{t^i},$$
when calculated with respect to the basis $1,t,\dots,t^{m-1}$, identifying an element $h=\sum_{i=0}^{n}a_it^i$
with the vector $(a_1,\dots,a_n)$.
\\ (iv) If $S_f$ has no zero divisors then $T_f$ is irreducible, i.e. $\{0\}$ and $S^m$ are the only
$T_f$-invariant left $S$-submodules of $S^m$.
\end{theorem}

\begin{proof}
This is proved for instance in \cite[Theorem 13 (2),  (3), (4)]{LS} for $\delta=0$, $f$ irreducible and $S$ a finite field.
 The proofs generalize easily and mostly verbatim to our more general setting.
\end{proof}

From Theorem \ref{thm:main1}  (vi) together with Theorem \ref{thm:leftmult} (i) we obtain:

\begin{corollary}
 Let $f(t)=\sum_{i=0}^{m}d_it^i\in S[t;\sigma]$ with $d_0$ invertible.
If $\sigma$ is not surjective then the pseudolinear transformation $T_f$ is not surjective.
  In particular, if $S$ is a division ring, $f$ irreducible and $\sigma$ is not  surjective then $T_f$ is not surjective.

Moreover, if $\sigma$ is bijective then $T_f$ is surjective.
\end{corollary}

\begin{remark}
(i) From Theorem \ref{thm:leftmult} we obtain \cite[Lemma 2]{BL}, since
$pq=0$ in $S_f$ is equivalent to $L_p(q)=p(T_f)=0$.
Note that
$$T_f^n(ah)=\sum_{i=0}^{n}\Delta_{i,n}(a)T_f^i(h)$$
 for all $a\in S$, $h\in S^m$
\cite{BL}, so  $L_{t^n}$ is usually not $(\sigma,\delta)$-pseudolinear anymore.
\\ (ii) Right multiplication with $h$ in $S_f$ induces the injective $S$-linear map
$$\gamma: S_f \to {\rm Mat}_m(S),\quad h\mapsto R_h \mapsto  Y.$$
 $f$ is two-sided is equivalent to
 $\gamma$ being the right regular representation of $S_f$. In that case, $\gamma$ is an injective ring
 homomorphism. In particular, (1) and (3)  in \cite[Theorem 6.6]{FG} hold in our general setting (i.e., for any choice of $f$)
  if and only if $S_f$ is associative:
  both reflect the fact that then $\gamma:S_f\longrightarrow {\rm Mat}_m(S)$ is the right regular representation of $S_f$.
\\ (iii) Suppose $f=h'g=gh$.
Right multiplication in $S_f$ induces the left $S$-module endomorphisms
$R_{h}$ and $R_g$. We have
$g\in {\rm ker}(R_{h})=\{ u\in R\,|\, {\rm deg}(u)<m \text{ and } uh\in Rf  \}$
and
$h\in {\rm ker}(R_g)=\{ u\in R\,|\,{\rm deg}(u)<m \text{ and } ug\in Rf  \}$, cf. \cite[Lemma 3]{L} or
  \cite[Theorem 6.6]{FG}.
If $f$ is two-sided, ${\rm ker}(R_g)=S_f h$ and ${\rm ker}(R_h)=S_f g$.
  \\ (iv) Suppose $f=h'g=gh$.
Left multiplication in $S_f$ induces the right $S_0$-module endomorphisms
$L_{h'}$ and $L_g$. We have
$g\in {\rm ker}(L_{h'})=\{ u\in R\,|\,{\rm deg}(u)<m \text{ and } hu\in Rf  \}$
and
$h\in {\rm ker}(L_g)=\{ u\in R\,|\, {\rm deg}(u)<m \text{ and }gu\in Rf  \}.$
If $f$ is two-sided, ${\rm ker}(L_{h'})= gS_f$ and ${\rm ker}(L_g)= h'S_f$.

Furthermore, (iii) and (iv) tie in with or generalize (4), (5) in \cite[Theorem 6.6]{FG}.
\end{remark}

%%%%%%%%%%%%%%%%%%%%%%%%%%%%%%%%%%%%%%%%%%%%%%%%%%%%%%%%%%%%%%%%%%%%%%%

\section{Finite nonassociative rings obtained from skew polynomials over finite chain-rings} \label{sec:FCRs}

%%%%%%%%%%%%%%%%%%%%%%%%%%%%%%%%%%%%%%%%%%%%%%%%%%%%%%%%%%%%%%%%%%%%%%%%%

\subsection{Finite Chain Rings (cf. for instance \cite{Mc})}
When $S$ is a finite ring, $S_f$ is a  finite  unital nonassociative ring with $|S|^m$ elements
 and
a finite unital nonassociative algebra  over the finite subring $S_0$ of $S$.
E.g., if $S$ is a finite field and $f$ irreducible, then $S_f$ is a finite unital nonassociative division ring, also
called a \emph{semifield} \cite{LS}.
We will look at the special case where $S$ is a finite chain ring.
Lately, these rings gained substantial momentum in coding theory, see for instance
\cite{BG}, \cite{B},  \cite{BSU08}, \cite{C}, \cite{FN},  \cite{GK},  \cite{KZM}, \cite{LL}.

A finite unital commutative ring $R\not=\{0\}$ is called a \emph{finite chain ring}, if its ideals are linearly
ordered by inclusion.

Every ideal of a finite chain ring is principal and its maximal ideal is unique. In particular, $R$ is a
local ring and the residue  field $K=R/(\gamma)$, where $\gamma$ is
a generator of its maximal ideal $m$, is a finite field.
 The ideals $(\gamma^i)=\gamma^i R$ of $R$ form the proper chain
$$R=(1)\supseteq(\gamma)\supseteq (\gamma^2) \supseteq \dots \supseteq(\gamma^e)=(0).$$
The integer $e$ is called the \emph{nilpotency index} of $R$.
 If $K$ has $q$ elements, then $|R|=q^e$. If $\pi:S\longrightarrow K=R/(\gamma)$,
 $x\mapsto \overline{x}=x \,{\rm mod}\, \gamma$ is the canonical projection,
 a monic polynomial $f\in R[t]$ is called \emph{base irreducible} if $f$ is irreducible in $K$.

Let $R$ and $S$ be two finite  chain rings such that $R\subset S$ and $1_R=1_S$
Then $S$ is an extension of $R$ denoted $S/R$. If $m$ is the maximal ideal of $R$ and $M$ the one of $S$, then
$S/R$ is called \emph{separable} if $mS=M$.
 The \emph{Galois group of} $S/R$ is the group $G$ of all automorphisms of $S$ which are the identity when restricted
 to $R$. A separable extension $S/R$ is called \emph{Galois} if
 $S^G=\{s\in S\,|\,\tau(s)=s \text{ for all }\tau\in G\}=R$. This is equivalent to
$S=R[x]/(f(x))$, where $(f(x))$ is the ideal generated by a monic basic irreducible polynomial $f(x)\in R[x]$
\cite[Theorem XIV.8]{Mc}, \cite[Section 4]{W}.
From now on, a separable extension $S/R$ of finite chain rings is understood to be a separable
Galois extension.

The Galois group $G$ of a separable extension
$S/R$ is isomorphic to the Galois group of the extension $\mathbb{F}_{q^n}/\mathbb{F}_{q}$,
where $\mathbb{F}_{q^n}=S/M$, $\mathbb{F}_{q}=R/m$. $G$ is cyclic with generator $\sigma(a)=a^q$ for a suitable
primitive element $a\in S$, and $\{a,\sigma(a),\dots,\sigma^{n-1}(a)\}$ is a free $R$-basis of $S$.
Since $S$ is also an unramified extension of $R$, $M=Sm=Sp$, and
$$S=(1)\supseteq Sp\supseteq\dots\supseteq Sp^t=(0).$$
The automorphism groups of $S$ are known \cite{A1, A2}.

\begin{example}
(i) The integer residue ring $\mathbb{Z}_{p^e}$ and the ring
$\mathbb{F}_{p^n}[u]/(u^e)$ are finite chain rings of characteristic $p$, the later has
nilpotency index $e$ and residue field $\mathbb{F}_{p^n}$.
\\ (ii)
 A finite unital ring $R$ is called a \emph{Galois ring} if it is commutative,
and its zero-divisors $\Delta(R)$ have the form $pR$ for some prime $p$.
$(p)=Rp$ is the unique maximal ideal of $R$.
Given a prime $p$ and positive integers $e$, $n$, denote by $G(p^e,n)$ the Galois ring of characteristic
$p^s$ and cardinality $p^{en}$ which is unique up to isomorphism.
 Its residue field (also called \emph{top-factor}) $\overline{ G(p^e,n)}=G(p^e,n)/pG(p^e,n)$
 is the finite field $\mathbb{F}_{p^n}$.
\end{example}

\subsection{Skew-polynomials and Petit's algebras over finite chain rings}

Let $S$ be a finite chain ring with residue class field  $K=S/(\gamma)$ and $\sigma\in {\rm Aut}(S)$,
$\delta$ a left $\sigma$-derivation. Consider the skew polynomial ring
$R=S[t;\sigma,\delta]$.
  Whenever $S$ is a finite chain ring, we suppose $\sigma((\gamma))= (\gamma)$
and  $\delta((\gamma))\subset(\gamma)$.
 Then the automorphism $\sigma$ induces an automorphism
$$\overline{\sigma}:K\rightarrow K,\,\, \overline{\sigma}(\overline{x})=\overline{\sigma(x)}$$
with $\sigma=\overline{\sigma} \circ \pi$, and analogously $\delta$ a left $\overline{\sigma}$-derivation
 $\overline{\delta}:K\rightarrow K$.
There is the canonical surjective ring homomorphism
$$\can:S[t;\sigma,\delta]\rightarrow K[t;\overline{\sigma},\overline{\delta}], \,\,g(t)=\sum_{i=0}^{n}a_it^i \mapsto \overline{g}(t)
=\sum_{i=0}^{n}\overline{a_i}t^i.$$
We call $f$ \emph{base irreducible}
 if $\overline{f}$ is irreducible in $K[t;\overline{\sigma},\overline{\delta}]$ and \emph{regular} if $\overline{f}\not=0$.
Obviously, if $\overline{f}$ is irreducible in $K[t;\overline{\sigma},\overline{\delta}]$ then $f$ is
irreducible in $S[t;\sigma,\delta]$.
Since $S_f\cong S_{af}$ for all invertible  $a\in S$, without loss of generality we consider only monic $f$ in this section.
From now on let $f\in R=S[t;\sigma,\delta]$ be monic of degree $m>1$.

\begin{lemma} \label{le:6}
Suppose $S$ is a finite chain ring with cardinality $q^e$.
 Then
 $$S_f=S[t;\sigma,\delta]/S[t;\sigma,\delta]f$$
  is a nonassociative finite ring with $q^{em}$ elements and
 $S_{\overline{f}}=K[t;\overline{\sigma},\overline{\delta}]/K[t;\overline{\sigma},\overline{\delta}]\overline{f}$
has $q^{m}$ elements.

 In particular, if $S=G(p^s,n)$ then
$S_f$ has $p^{snm}$ elements and $S_{\overline{f}}$ has $p^{nm}$ elements.
\end{lemma}

\begin{proof}
 The residue class field $K$ has $q$ elements if $|S|=q^e$.
 Since $S_f$ is a left $S$-module with basis $t^i$, $0\leq i\leq m-1$, it has $q^{em}$ elements, analogously,
 $S_{\overline{f}}$ has $q^{m}$ elements.
\end{proof}

From  Remark \ref{rem:7}, Proposition \ref{prop:nuctest}, \cite[(9)]{P66} and \cite[(7)]{P66} we get (as all polynomials in
$K=K[t;\overline{\sigma}]$ are bounded for
 a finite field $K$, and $K[t;\overline{\sigma},\overline{\delta}]\cong K[t;\sigma']$ for a suitable $\sigma'$):

\begin{corollary} \label{cor:FCRmain1}
Suppose $S$ is a finite chain ring.
\\ (i) $S_f$ is a unital nonassociative algebra with finitely many elements  over the subring
$S_0=\{a\in S\,|\, ah=ha \text{ for all } h\in S_f\}$ of $S$.
\\ (ii)
 $S_{\overline{f}}=K[t;\overline{\sigma},\overline{\delta}]/\overline{f}K[t;\overline{\sigma},\overline{\delta}]$
 is a semifield if and
 only if $f$ is base irreducible, if and only if ${\rm Nuc}_r(S_{\overline{f}})$ has no zero divisors.
 \\ (iii) If $\delta=0$ then ${\rm Fix}(\sigma)\subset S_0$.
\end{corollary}

From now on we
assume that $\gamma\in {\rm Fix}(\sigma)\cap {\rm Const}(\delta)$.
 Then $\gamma S_f$ is a two-sided ideal in $S_f$.

 The canonical surjective ring homomorphism
$\can:S[t;\sigma,\delta]\rightarrow K[t;\overline{\sigma},\overline{\delta}]$
 induces the surjective homomorphism of nonassociative rings
 $$\Psi:S_f=S[t;\sigma,\delta]/S[t;\sigma,\delta]f \rightarrow
K[t;\overline{\sigma},\overline{\delta}]/K[t;\overline{\sigma},\overline{\delta}]\overline{f},$$
$$g(t) \mapsto \overline{g}(t)$$
which has as kernel the two-sided ideal $\gamma S_f$.

This induces an isomorphism of nonassociative rings:
$$ (1)\quad\quad  S_f/\gamma S_f\cong
K[t;\overline{\sigma},\overline{\delta}]/K[t;\overline{\sigma},\overline{\delta}]\overline{f}= S_{\overline{f}},$$
$$g(t)+ \gamma S_f \mapsto \overline{g}(t).$$

%%%%%%%%%%%%%%%%%%%%%%%%%%%%%%%%%%%%%%%%%%%%%%%%%%%%%%%%%%%%%%%%%%%%%%%%%%%
\subsection{Generalized Galois rings}
%%%%%%%%%%%%%%%%%%%%%%%%%%%%%%%%%%%%%%%%%%%%%%%%%%%%%%%%%%%%%%%%%%%%%%%%%%%%%

 A \emph{generalized Galois ring} (GGR) is a finite nonassociative unital ring
 $A$ such that the set of its (left or right) zero divisors $\Delta(A)$ has the form $pA$ for some prime $p$.
$\Delta(A)$ is a two-sided ideal and the quotient $\overline{A}= A/pA$ is a semifield of characteristic $p$,
called the \emph{top-factor} of $A$. The characteristic of $A$ is $p^s$. There is a canonical epimorphism
$$ A\longrightarrow \overline{A}=A/pA, \quad a \mapsto\bar a=a+pA.$$
A generalized Galois ring $A$ of characteristic $p^s$ is a \emph{lifting} of the semifield $\overline{A}$
\emph{of characteristic} $p^s$
if $\overline{C(A)}=C(A)/pC(A)\cong C(\overline{A})$ (cf. \cite{Cons}).

A finite unital ring $A$ is a GGR if and only if there is a prime $p$ and a positive integer $s$ such that
${\rm char}(A)=p^s$ and $\overline{A}=A/pA$ is a semifield \cite[Theorem 1]{Cons}.

Let  $S=G(p^e,n)$ be a Galois ring and let $f\in R=S[t;\sigma,\delta]$ be monic of degree $m>1$ as before.

Let $A=S_f=S[t;\sigma,\delta]/S[t;\sigma,\delta]f$, then by (1)
there is the canonical isomorphism
$$   A/p A\cong
K[t;\overline{\sigma},\overline{\delta}]/K[t;\overline{\sigma},\overline{\delta}]\overline{f}= S_{\overline{f}}.$$
Thus all  base irreducible such $f\in S[t;\sigma,\delta]$ yield generalized Galois rings $S_f$:

\begin{theorem}
Let $S$ be a Galois ring and let $f(t)\in S[t;\sigma,\delta]$ be base irreducible. Then the finite nonassociative ring
$$S_f=S[t;\sigma,\delta]/S[t;\sigma,\delta]f$$
 is a GGR with $p^{enm}$ elements. If $S_f$  is not associative it is a lifting of its top-factor since
$S_0/pS_0\cong {\rm Fix}(\overline{\sigma})$.
\end{theorem}

\begin{proof}
If $\overline{f}$ is irreducible, then
$S_{\overline{f}}=K[t;\overline{\sigma},\overline{\delta}]/K[t;\overline{\sigma},
\overline{\delta}]\overline{f}$ is a semifield. By  (1), we have
 $S_{\overline{f}}\cong A/pA=\overline{A}$, so that $\overline{A}$ is a semifield.
Thus
$S_f$ is a GGR with $p^{enm}$ elements by Lemma \ref{le:6} and \cite[Theorem 1]{Cons}.

 Every left $\sigma$-derivation of a finite field is inner, so that there are a suitable $y$ and
$\widetilde{f}\in K[y;\overline{\sigma}]$ such that
$S_{\overline{f}}\cong K[y;\overline{\sigma}]/K[y;\overline{\sigma}]\widetilde{f}$.
 The second assertion is now proved using
 the fact that $S_{\widetilde{f}}$ is a semifield over ${\rm Fix}(\overline{\sigma})$ by Theorem \ref{thm:main2} (ii)
 and that $\overline{C(A)}=C(A)/pC(A)\cong C(\overline{A})$ .
\end{proof}

\begin{corollary} \label{cor:18}
Let $S/S_0$ be a Galois extension of Galois rings with Galois group
${\rm Gal}(S/S_0)=\langle\sigma\rangle$  of order $m$ and let $F$ denote the residue field of $S_0$, $ char (F)=p$.
Choose $f(t)=t^m+ph(t)-d\in R=S[t;\sigma]$ with $d\in S\setminus S_0$ invertible and $h(t)\in S[t;\sigma]$ of degree $<m$.
\\ (i) If the elements $1,\overline{d},\dots,\overline{d}^m$ are linearly independent over $F$, then
$S_f$ is a GGR which is a lifting of its top-factor.
\\ (ii) For every prime $m$, $S_f$ is a GGR which is a lifting of its top-factor.
\end{corollary}

\begin{proof}
$K/F$ is a Galois extension with Galois group ${\rm Gal}(K/F)=\langle\overline{\sigma}\rangle$ of order $m$.
We have $\overline{f}(t)=t^m-\overline{d}$. With the assumptions in (i) resp. (ii),
$S_{\overline{f}}$ is a nonassociative cyclic division algebra over $F$ \cite{S12} and thus the finite nonassociative ring
$S_f$ is a GGR by \cite[Theorem 1]{Cons}.
 It is straightforward to see  that ${\rm Fix}(\sigma)={\rm Fix}(\overline{\sigma})$ using isomorphism (1) and
 that $S_f$ is a lifting of its top-factor by Theorem \ref{thm:main1}.
\end{proof}

Note that although the top-factor in Corollary \ref{cor:18} is a nonassociative cyclic algebra,
 it is unlikely that the algebra $S_{f}$ is isomorphic to a generalized nonassociative cyclic algebra as defined
  in Example \ref{ex:cyclic} unless $h=0$.

%%%%%%%%%%%%%%%%%%%%%%%%%%%%%%%%%%%%%%%%%%%%%%%%%%%%%%%%%%%%%%%%%%%%%%%%%%%
%
%Linear codes over finite  chain rings
%
%%%%%%%%%%%%%%%%%%%%%%%%%%%%%%%%%%%%%%%%%%%%%%%%%%%%%%%%%%%%%%%%%%%%%%%%%%%

\section{Linear codes} \label{sec:codes}

%%%%%%%%%%%%%%%%%%%%%%%%%%%%%%%%%%%%%%%%%%%%%%%%%%%%%%%%%%%%%%%%%%%%%%%%%%%
\subsection{Cyclic $(f,\sigma,\delta)$-codes}
%%%%%%%%%%%%%%%%%%%%%%%%%%%%%%%%%%%%%%%%%%%%%%%%%%%%%%%%%%%%%%%%%%%%%%%%%%%

A \emph{linear code of length $m$ over $S$} is a submodule of the $S$-module $S^m$.
From now on, let $f\in S[t;\sigma,\delta]$ be a  monic polynomial of degree $m>1$.

A \emph{cyclic $(f,\sigma,\delta)$-code} $\mathcal{C}\subset S^m$ is a subset of $S^m$ consisting of the vectors
$(a_0,\dots,a_{m-1})$ obtained from elements $h=\sum_{i=0}^{m-1}a_it^i$
in a left principal ideal $gS_f=S[t;\sigma,\delta]g/S[t;\sigma,\delta]f$ of $S_f$, with $g$ a monic right divisor of $f$.

 A code $\mathcal{C}$ over $S$ is called
   $\sigma$-\emph{constacyclic} if $\delta=0$ and there is a non-zero $d\in S$ such that
   $$(a_0,\dots,a_{m-1})\in  \mathcal{C}\Rightarrow (\sigma(a_{m-1})d,\sigma(a_0),\dots,\sigma(a_{m-2}))\in  \mathcal{C}.$$
If $d=1$, the code is called \emph{$\sigma$-cyclic}.

\cite[Theorem 1]{BL}, the first three equivalences of
\cite[Theorem 2]{BL} and \cite[Corollary 1]{BL} translate  to our set-up as follows (the first equivalences in
\cite[Theorem 2]{BL} are now trivial):

\begin{theorem}\label{thm:newTheorem3.2}
Let $g=\sum_{i=0}^{r}g_it^i$ be a monic polynomial which is a right divisor of $f$.
\\ (i) The cyclic $(f,\sigma,\delta)$-code $\mathcal{C}\subset S^m$ corresponding to the principal ideal
$g S_f$ is a free left $S$-module of dimension $m-{\rm deg}g$.
\\ (ii) If $(a_0,\dots,a_{m-1})\in \mathcal{C}$ then $L_t(a_0,\dots,a_{m-1})\in \mathcal{C}$.
\\ (iii) The matrix generating $\mathcal{C}$ represents the right multiplication $R_g$ with $g$ in $S_f$,
calculated with respect to the basis $1,t,\dots,t^{m-1}$, identifying elements $h=\sum_{i=0}^{m-1}a_it^i$
with the vectors $(a_0,\dots,a_{m-1})$.
\end{theorem}

Note that (iii) is a straightforward consequence from the fact that the $k$-th row of the matrix
generating $\mathcal{C}$ is given by left multiplication of $g$ with $t^k$
in $S_f$, i.e. by
$$L_{t^k}(g)=L_t^k(g).$$
In particular, when $\delta=0$ and $f(t)=t^m-d$, for any $p\in S_f$, the matrix representing
right multiplication $R_p$ with respect to the basis $1,t,\dots,t^{m-1}$ is the circulant matrix defined in
\cite[Definition 3.1]{FG}, see also Section \ref{sec:maps}.

\begin{theorem}\label{thm:newTheorem3.6}
Let $g=\sum_{i=0}^{r}g_it^i$ be a monic polynomial which is a right divisor of $f$, such that $f=gh=h'g$ for two monic
polynomials $h,h'\in S[t;\sigma,\delta]$. Let $\mathcal{C}$ be the cyclic $(f,\sigma,\delta)$-code corresponding to $g$ and
$c=\sum_{i=0}^{m-1}c_it^i\in S[t;\sigma,\delta]$.
Then the following are equivalent:
\\ (i) $(c_0,\dots,c_{m-1})\in \mathcal{C}$.
\\ (ii) $c(t)h(t)=0$ in $S_f$.
\\ (iii) $L_c(h)=ch=0$, resp. $R_h(c)=hc=0$.
\end{theorem}

This is already part of \cite[Theorem 2]{BL}
 and generalizes \cite[Proposition 1]{DO}: it shows that sometimes $h$ is a parity check polynomial for $\mathcal{C}$
 also when $f$ is not two-sided. Note that when we only have
  $hg=f$, $h$ monic, and $\mathcal{C}$ is the code generated by $g$ then if $ch=0$ in $S_f$, $c$ is a codeword of $\mathcal{C}$.

\begin{corollary} \label{cor:newCor3.6}
Let $g=\sum_{i=0}^{r}g_it^i$ be a monic polynomial which is a right divisor of $f$, such that $f=gh=h'g$ for two monic
polynomials $h,h'\in S_f$. Let $\mathcal{C}$ be the cyclic $(f,\sigma,\delta)$-code corresponding to $g$.
Then the matrix representing right multiplication $R_h$ with $h$ in $S_f$ with respect to
 the basis $1,t,\dots,t^{m-1}$
is a control matrix of the cyclic $(f,\sigma,\delta)$-code corresponding to $g$.
\end{corollary}

\begin{proof}
 The matrix $H$ with $i$th row the vector representing
$$L_{t^{i-1}}(h)=t^{i-1}h,$$
$1\leq i\leq m$, is the matrix representing right multiplication $R_h(p)=ph$ with $h$ in $S_f$ with respect to
 the basis $1,t,\dots,t^{m-1}$, since $t^{i-1}h=R_h(t^{i-1})$ is the $i$th row.
\end{proof}

 For a linear code $\mathcal{C}$ of length $m$ we denote by $\mathcal{C}(t)$ the set of skew polynomials
 $a(t)=\sum_{i=0}^{m-1}a_it^i\in S_f$ associated to the codewords $(a_0,\dots,a_n)\in \mathcal{C}$.

As a consequence of Proposition \ref{prop:skewcodemain} and Theorem \ref{thm:newTheorem3.2} we obtain a description of
$\sigma$-constacyclic codes in terms of  left ideals of $S_f$, generalizing \cite[Theorem 2.2]{JL}:

\begin{corollary} \label{thm:skewcodemain}
Let $f=t^m-d\in S[t;\sigma]$, $d\in S$ invertible, and $\mathcal{C}$ a linear code over $S$ of length $m$.
\\ (i)  Every  left  ideal of $S_f$ with $f=t^m-d\in S[t;\sigma]$ generated by a monic right divisor $g$ of
$f$ in $S[t;\sigma]$ yields a $\sigma$-constacyclic code
  of length $m$ and dimension $m-{\rm deg} g$.
\\ (ii) If $\mathcal{C}$ is a $\sigma$-constacyclic code
 then the skew polynomials in the set
$\mathcal{C}(t)$ of elements $a(t)$ obtained from $(a_0,\dots,a_{m-1})\in \mathcal{C}$
form a left ideal of $S_f$ with $f=t^m-d\in S[t;\sigma]$.
\end{corollary}

\begin{proof}
(i) follows from Theorem \ref{thm:newTheorem3.2}.
\\ (ii) The argument is analogous to the proof of \cite[Theorem 1]{BGU07}.
\end{proof}

For any monic $f\in S[t;\sigma,\delta]$,
 representing the right multiplication $R_g$ in $S_f$  by the matrix $Y$ calculated with
respect to the $S$-basis
$1,t,\dots,t^{m-1}$ gives the injective $S$-linear map
$$\gamma: S_f \to {\rm Mat}_m(S), \quad h\mapsto R_h \mapsto Y.$$
For algebras $S_f$ which are not associative, this is not a regular representation of the algebra. However,
we can prove some weaker results for special choices of $f$:

\begin{lemma}\label{le:semi-multiplicative}
Suppose that $f(t)=t^m-d_0\in S[t;\sigma,\delta]$ or $f(t)=t^2-d_1t-d_0\in S[t;\sigma,\delta]$.
Then the product of the $m\times m$ matrices representing $R_d$, $0\not=d\in S\subset S_f$,
and  $R_g$ for any $0\not=g\in S_f$, is the matrix representing $R_{dg}$, i.e. the matrix
representing the right multiplication with $dg$ in $ S_f$.
\end{lemma}

The proofs are straightforward but tedious calculations \cite{CB}. The case where
$f(t)=t^m-d_0\in S[t;\sigma]$ and $S$ is a  cyclic Galois extension of degree $m$ over a field $F$ with
$\sigma$ generating its automorphism group is already treated in \cite{S13}, its proof holds analogously when $S$
is a commutative ring with an automorphism $\sigma$ of order $m$.

When $S$ is a commutative unital ring, we define a map $M: S_f \rightarrow S$ by
\[M(h) = \det(\gamma(h))\]
for all $h \in S_f$. Note that this is analogous to the definition of the reduced norm of an associative central
simple algebra.

We recall the following:
Let $A$ be an algebra over a ring $S_0$ and $D$ a subalgebra of $A$, both free of finite rank as $S_0$-modules.
Then a map $M:A\mapsto D$  of degree $n$ is called \emph{left semi-multiplicative} if
$$M(ax)=M(a)M(x) \text{ for all } a\in D, x\in A.$$
Furthermore, a map $M:A\mapsto D$  has {\it degree $n$} over $S_0$
if $M(a v)=a^n M(v)$ for all $a\in S_0$, $v\in A$ and if the map $M : A \times
\dots \times A \mapsto D$ defined by
 $$M(v_1,\dots,v_n)=
  \sum_{1\leq i_1< \dots<i_l\leq n}(-1)^{n-l}M(v_{i_1}+ \dots +v_{i_l})$$
($1\leq l\leq n$)  is an $n$-linear map over $S_0$, i.e., $M : A \times\dots \times A \mapsto D$ ($n$-copies) is an
$S_0$-multilinear map where $M(v_1,\dots, v_n)$ is invariant under all permutations of its variables.

\begin{corollary} Suppose $S$ is a commutative unital ring and both $S$ and the algebra $S_f$ are free of finite rank as $S_0$-module.
For $f(t)=t^m-d_0\in S[t;\sigma,\delta]$ or $f(t)=t^2-d_1t-d_0\in S[t;\sigma,\delta]$, the map
$$M: S_f \rightarrow S,\quad M(h) = \det(\gamma(h)),$$
is left semi-linear of degree $m$.
\end{corollary}

This is a direct consequence of Lemma \ref{le:semi-multiplicative}.
For properties of left semi-linear maps, especially for those of lower degree, the reader is referred to \cite{S13}, \cite{S13.1}.

\begin{example}
Let $K/F$ be a cyclic Galois extension of degree $m$ with reduced norm $N_{K/F}$ and reduced trace $T_{K/F}$,
${\rm Gal}(K/F)=<\sigma>$ and $f(t)=t^m-d\in K[t;\sigma]$.
Then $M: S_f \rightarrow S$  is a left semi-multiplicative map of degree $m$.
 If $a\in K$ is considered as an element of $S_f$ then $M(a) = N_{K/F}(a)$.
In particular, for $m=3$ and $h=h_0+h_1t+h_2t^2$, we have
$$M(h)=N_{K/F}(h_0)+dN_{K/F}(h_1)+d^2N_{K/F}(h_2)-
dT_{K/F}(h_0h_1h_2)$$
  \cite{S13}.
\end{example}

\begin{remark}
We point out that if $S=(K/F,\varrho,c)$ is a
suitable cyclic division algebra with norm $N_{S/F}$, we can describe the right multiplication with $h$
by an $mn\times mn$ matrix $X(h)$ with entries in $K$ as described in Remark \ref{re:I} (ii), and
define a map
$$M:S_f\longrightarrow S, \quad M(h)=\det(R_h)=\det (X(h))$$
 which is also left-semilinear for
suitable $f(t)=t^m-d$
(cf. \cite[Remark 19]{Pu16.1} where we look at the matrix representing left multiplication instead, since we are
dealing with the opposite algebra
there).
Again the map  $M$
 can be seen as a generalization
 of the norm of an associative central simple algebra and
 $$M(x)=N_{F/S_0}(N_{S/F}(x))$$
 for all $x\in S$ for suitably chosen $S_0$-algebras $S_f$, for details see  \cite{Pu16.1}.
\end{remark}

%%%%%%%%%%%%%%%%%%%%%%%%%%%%%%%%%%%%%%%%%%%%%%%%%%%%%%%%%%%%%%%%%%%%%%%%%%%%%
\subsection{Codes over finite  chain rings}
%%%%%%%%%%%%%%%%%%%%%%%%%%%%%%%%%%%%%%%%%%%%%%%%%%%%%%%%%%%%%%%%%%%%%%%%%%%%%%%

Let $S$ be a finite  chain ring and $\sigma$ an automorphism of $S$.
The $S[t;\sigma]$-module $S[t;\sigma]/S[t;\sigma]f$ is increasingly favored
 for linear code constructions over $S$,
with $f$ a monic polynomial of degree $m$ (usually $f(t)=t^m-d$), cf. for instance
\cite{B}, \cite{BSU08}, \cite{JL}. For code constructions, we generally look at reducible skew polynomials $f$.

We take the setup discussed in  \cite{B}, \cite{BSU08}, \cite{JL},
where the $S[t;\sigma]$-module $S[t;\sigma]/S[t;\sigma]f$ is employed for linear code constructions,
 and discuss on some examples how the results mentioned previously fit into our view of
equipping  $S[t;\sigma]/S[t;\sigma]f$ with a nonassociative algebra structure:

\begin{itemize}
\item  In  \cite[Theorem 2.2]{JL}, it is shown that a code of length $n$ is $\sigma$-constacyclic if and only if the skew polynomial
representation
associated to it is a left ideal in $S_f$, again assuming $S_f$ to be associative, i.e. $f(t)=t^m-d\in S[t;\sigma]$ with
$d\in S$ invertible, to be two-sided, and $S$ to be a finite chain ring.

\item  In \cite[Proposition 2.1]{BSU08}, it is shown that any right divisor $g(t)$ of $f(t)=t^m-d\in S[t;\sigma]$
generates a  principal left ideal
in $S_f$, provided that $f$ is a monic two-sided element and assuming $S$ is a Galois ring.
The codewords associated with the elements in the ideal $Rg$ form a code of length $m$ and dimension $m-{\rm deg} g$.
This also holds in the nonassociative setting, so we can drop the assumption in \cite[Proposition 2.1]{BSU08} that $f$ needs to be a monic central element,
see Corollary \ref{thm:skewcodemain}.

\item In  \cite[Theorem 2]{B} (or similarly in  \cite[3.1]{JL}), it is shown that if a skew-linear code $\mathcal{C}$ is associated with a principal left ideal, then
$\mathcal{C}$ is an $S$-free module if and only if $g$ is a right divisor of $f(t)=t^m-1$, again assuming $S$ to be Galois,
and $f$ two-sided. This is generalized in Proposition \ref{prop:skewcodemain}, resp. Corollary \ref{thm:skewcodemain}.

\item For $f(t)=t^m-d\in \mathbb{F}_{q}[t;\sigma]$, the \emph{$(\sigma,d)$-circulant matrix} $M_d^\sigma$  in \cite{FG} is the matrix
representing $R_g$ in the algebra $ S_f$ calculated with respect to the basis
$1,t,\dots,t^{m-1}$. Therefore \cite[Theorem 3.6]{FG} states that for associative algebras $S_f$,
right multiplication gives the right regular  representation of the algebra, so that the product of the matrix representing
$R_h$,
and the one representing $R_g$, for any $0\not=h\in S_f$, $0\not=g\in S_f$, is the matrix
representing $R_{hg}$ in $ S_f$.
The fact that $\gamma$ is injective and additive is observed in \cite[Remark 3.2 (a)]{FG}.

Lemma \ref{le:semi-multiplicative} and the fact that $\gamma$ is $S$-linear imply \cite[Remark 3.2 (b)]{FG}.

Moreover, the matrix equation in \cite[Theorem 5.6 (1)]{FG} can be read as
follows: if $t^n-a=hg$ and $c=\gamma(a,g)$, then the matrix
representing the right multiplication with the element $g(t)\in R_n$ in the algebra $S_{f}$ where
$f(t)=t^n-a\in \mathbb{F}_{q}[t;\sigma]$,
equals  the transpose of the matrix
representing the right multiplication with an element $g^{\sharp}(t)\in S_{f_1}$ where
$f_1(t)=t^n-c^{-1}\in \mathbb{F}_{q}[t;\sigma]$. This suggests  an isomorphism between
$S_{f_1}=\mathbb{F}_{q}[t;\sigma]/\mathbb{F}_{q}[t;\sigma]f_1$ and the opposite algebra of
$S_{f}=\mathbb{F}_{q}[t;\sigma]/\mathbb{F}_{q}[t;\sigma]f$.
\end{itemize}

\section{Conclusion and further work}

This paper proposes a more general way of looking at cyclic $(f,\sigma,\delta)$-codes using nonassociative algebras, and
 unifies different ways of designing cyclic linear $(f,\sigma,\delta)$-codes
 in a general, nonassociative theory. Connections between the algebras and some fast-decodable space-time block code designs are pointed out along the way.

It is well known that for any $f\in R=S[t;\sigma,\delta]$, $R/Rf$ is an $R$-module with the module structure given by
the multiplication $g(h+Rf)=gh+Rf=r+Rf$ if $r$ is the reminder of $gh$ after right dividing by $f$.
This is exactly the multiplication which makes the additive group $\{g\in R\,|\, {\rm deg} (g)<m \}$
 into a nonassociative algebra when $f$ has an invertible leading coefficient.
 Thus one might argue that the introduction of the nonassociative
 point of view we suggested here seems to make things only more complicated that actually needed and not necessarily
 better.

  The full benefits of this approach for coding theory
might only become visible once more work has been done in this direction.  Using the nonassociative Petit
algebras $S_f$
over number fields allows us for instance
to show how certain cyclic $(f,\sigma,\delta)$-codes over finite rings canonically induce a
$\mathbb{Z}$-lattice in $\mathbb{R}^N$. The observations in
 \cite[Section 5.2, 5.3]{DO} hold analogously for our nonassociative algebras and
 explain the potential of the algebras $S_f$
  for coset coding in space-time block coding, in particular for wiretap coding, cf. \cite{Pu15}.
 Previous results for lattices obtained from
 $\sigma$-constacyclic codes related to associative cyclic algebras  by Ducoat and Oggier \cite{DO}
  are obtained as special cases.

We also canonically obtain coset codes from orders in nonassociative algebras over number fields which are
used for fast-decodable space-time block codes  \cite{Pu16}.
Again, previous results for coset codes related to associative cyclic algebras $S_f$ by Oggier and Sethuraman
 \cite{OS} are obtained as special cases.

%*******************************************************************************************%
%****************************************************************************************%

\end{document}